\documentclass{llncs}
\usepackage{spcl}
\usepackage{hyperref}
\newcommand{\doi}[1]{{doi:\href{http://doi.org/#1}{\nolinkurl{#1}}}\rmFullStop}

\newcommand{\eprintlnk}[1]{{\href{#1}{Electronic version}}\rmFullStop}
\renewcommand{\url}[1]{{\href{#1}{\nolinkurl{#1}}}\rmFullStop}
\newcommand*{\rmFullStop}{\rmifnextchar{.}{}{}}
\makeatletter
\newcommand{\rmifnextchar}[3]{%
  \begingroup
  \ltx@LocToksA{\endgroup#2}%
  \ltx@LocToksB{\endgroup#3}%
  \ltx@ifnextchar{#1}{%
    \def\next{\the\ltx@LocToksA}%
    \afterassignment\next
    \let\scratch= %
  }{%
    \the\ltx@LocToksB
  }%
}
\makeatother

\title{On the Strongest Three-Valued Paraconsistent Logic Contained in 
       Classical Logic and Its Dual}
\author{C.A. Middelburg}
\institute{Informatics Institute, Faculty of Science, University of
           Amsterdam, \\
           Science Park~904, 1098~XH Amsterdam, the Netherlands \\
           \email{C.A.Middelburg@uva.nl}}

\begin{document}
\maketitle

\begin{abstract}
\LPif\ is a three-valued paraconsistent propositional logic which is 
essentially the same as J3.
It has most properties that have been proposed as desirable properties 
of a reasonable paraconsistent propositional logic.
However, it follows easily from already published results that there are 
exactly 8192 different three-valued paraconsistent propositional logics 
that have the properties concerned.
In this paper, properties concerning the logical equivalence relation of 
a logic are used to distinguish \LPif\ from the others.
As one of the bonuses of focussing on the logical equivalence relation, 
it is found that only 32 of the 8192 logics have a logical equivalence 
relation that satisfies the identity, annihilation, idempotent, and 
commutative laws for conjunction and disjunction.
For most properties of \LPif\ that have been proposed as desirable 
properties of a reasonable paraconsistent propositional logic, its 
paracomplete analogue has a comparable property.  
In this paper, properties concerning the logical equivalence relation of 
a logic are also used to distinguish the paracomplete analogue of 
\LPif\ from the other three-valued paracomplete propositional logics 
with those comparable properties. 
\begin{keywords} 
paraconsistent logic, three-valued logic, logical consequence, logical 
equivalence, paracomplete logic. 
\end{keywords}
\end{abstract}

\section{Introduction}
\label{sect-intro}

A set of propositions is contradictory if there exists a proposition 
such that both that proposition and the negation of that proposition are 
logical consequences of it.
In classical propositional logic, every proposition is a logical 
consequence of every contradictory set of propositions.
In a paraconsistent propositional logic, this is not the case.

\LPif\ is the three-valued paraconsistent propositional logic 
LP~\cite{Pri79a} enriched with an implication connective for which the 
standard deduction theorem holds and a falsity constant. 
This logic, which is essentially the same as J3~\cite{DOt85a}, the 
propositional fragment of CLuNs~\cite{BC04a} without bi-implication, and 
LFI1~\cite{CCM07a}, has most properties that have been proposed as 
desirable properties of a reasonable paraconsistent propositional logic.
However, it follows easily from results presented 
in~\cite{AA15a,AAZ11b,CCM07a} that there are exactly 8192 different 
three-valued paraconsistent propositional logics that have the 
properties concerned.
In this paper, properties concerning the logical equivalence relation of 
a logic are used to distinguish \LPif\ from the others.

It turns out that only 32 of those 8192 logics are logics of which the 
logical equivalence relation satisfies the identity, annihilation, 
idempotent, and commutative laws for conjunction and disjunction; and
only 16 of them are logics of which the logical equivalence relation 
additionally satisfies the double negation law.
\LPif\ is one of those 16 logics.
Two additional classical laws of logical equivalence turn out to be 
sufficient to distinguish \LPif\ completely from the others.

The desirable properties of a reasonable paraconsistent propositional 
logic referred to above concern the logical consequence relation of a 
logic.
It does not follow from those properties that the identity, 
annihilation, idempotent, and commutative laws for conjunction and 
disjunction are satisfied by the logical equivalence relation of the 
logic.
Therefore, if closeness to classical propositional logic is considered 
important, it should be a desirable property of a reasonable 
paraconsistent propositional logic to have a logical equivalence 
relation that satisfies these classical laws of logical equivalence.
This would reduce the potentially interesting three-valued 
paraconsistent propositional logics from 8192 to 32.

The dual notion of paraconsistency is paracompleteness.
A paracomplete propositional logic is a propositional logic in which not 
every disjunction of a proposition and the negation of that proposition 
is a logical consequence of every set of propositions.
I have coined the name \KLif\ for the paracomplete analogue of \LPif.
This logic is considered to be the dual of \LPif.
For most properties of \LPif\ that have been proposed as desirable 
properties of a reasonable paraconsistent propositional logic, \KLif\ 
has a comparable property.  
Those comparable properties are properties that any reasonable 
paracomplete propositional logic should have.
In this paper, properties concerning the logical equivalence relation of 
a logic are also used to distinguish \KLif\ from other three-valued 
paracomplete propositional logics with those comparable properties.

\KLif\ is essentially the same as the propositional fragment of 
LPF~\cite{BCJ84a,Che86a} without the constant that represents the truth 
value that is interpreted as neither true nor false.
That is, LPF has a definedness connective, $\mathrm{\Delta}$, instead of
the implication connective of \KLif, but these connectives can be 
defined in terms of each other (see e.g.\ Section~3.1.2 
of~\cite{Avr91a}).
LPF is well-known in the area of formal methods for software 
development.
It is basic to formal specification and verified design in the software 
development method VDM~\cite{Jon90a}.

In~\cite{BM15b}, a process algebra is presented that allows for dealing 
with contradictory states.
In order to allow for this, the process algebra concerned is built on 
\LPif. 
During the search for a paraconsistent propositional logic on which such 
a process algebra can be built, satisfaction of certain classical laws 
of logical equivalence turned out to be essential.
\LPif\ is one of only four three-valued paraconsistent propositional 
logics with all of the desirable properties referred to above of which 
the logical equivalence relation satisfies the classical laws concerned.
This finding triggered the more elaborate work on the logical 
equivalence relations of three-valued paraconsistent propositional 
logics presented in the current paper.

The structure of this paper is as follows.
First, a survey of the paraconsistent propositional logic 
\LPif\ is given (Section~\ref{sect-LP-iimpl-false}).
Next, the known properties of \LPif\ that have been proposed as 
desirable properties of a reasonable paraconsistent propositional logic 
are discussed (Section~\ref{sect-properties}).
Then, properties concerning the logical equivalence relation of a logic 
are used to distinguish \LPif\ from the other three-valued 
paraconsistent propositional logics with the properties discussed 
earlier (Section~\ref{sect-characterization}).
After that, the logical equivalence relation of \LPif\ is examined 
further and its key role in the algebraization of \LPif\ is shown 
(Section~\ref{sect-logical-equiv}).
Thereafter, \KLif\ is introduced (Section~\ref{sect-dual}) and 
properties of \KLif\ are presented that are comparable to properties of 
\LPif\ that have been presented in the preceding sections 
(Section~\ref{sect-dual-properties}).
Finally, some concluding remarks are made (Section~\ref{sect-concl}).

It is relevant to realize that the work presented in this paper is 
restricted to three-valued paraconsistent propositional logics that are 
\emph{truth-functional} three-valued logics.

There is some overlap between this paper and~\cite{BM15b}.
This paper primarily generalizes and elaborates Section~2 of that paper 
in such a way that it may be of independent importance to the area of 
paraconsistent logics.

\section{The Paraconsistent Logic \LPif}
\label{sect-LP-iimpl-false}

A set of propositions $\Gamma$ is contradictory if there exists a 
proposition $A$ such that both $A$ and $\Not A$ is a logical consequence 
of $\Gamma$.
In classical propositional logic, every proposition is a logical 
consequence of a contradictory set of propositions.
Informally, a paraconsistent propositional logic is a propositional 
logic in which not every proposition is a logical consequence of every 
contradictory set of propositions.

More precisely, a propositional logic $\mathcal{L}$ is a 
\emph{paraconsistent} propositional logic if 
(a)~its logical consequence relation $\LCon_{\mathcal{L}}$ satisfies the 
condition that there exist formulas $A$ and $B$ of $\mathcal{L}$ such 
that $A, \Not A \nLCon_{\mathcal{L}} B$ and 
(b)~its negation connective $\Not$ satisfies the condition that, 
for each propositional variable $p$, both 
$p \nLCon_{\mathcal{L}} \Not p$ and $\Not p \nLCon_{\mathcal{L}} p$.

In~\cite{Pri79a}, Priest proposed the paraconsistent propositional logic
LP (Logic of Paradox).
The logic introduced in this section is LP enriched with a falsity 
constant and an implication connective for which the standard deduction 
theorem holds. 
This logic, called \LPif, is in fact the propositional 
fragment of CLuNs~\cite{BC04a} without bi-implications.

\LPif\ has the following logical constants and connectives:
a falsity constant $\False$,
a unary negation connective $\Not$, 
a binary conjunction connective $\And$, 
a binary disjunction connective $\Or$, and
a binary implication connective $\Impl$.
Truth and bi-implication are defined as abbreviations:
$\True$ stands for $\Not \False$ and
$A \BImpl B$ stands for $(A \Impl B) \And (B \Impl A)$.

A Hilbert-style deductive system for \LPif\ is given in
Table~\ref{proofsystem-LPiimpl}.
\begin{table}[!tb]
\caption{Hilbert-style deductive system for \LPif}
\label{proofsystem-LPiimpl}
\begin{eqntbl}
\begin{eqncol}
\mathbf{Axiom\; Schemas:}
\\
A \Impl (B \Impl A)
\\
(A \Impl (B \Impl C)) \Impl ((A \Impl B) \Impl (A \Impl C))
\\
((A \Impl B) \Impl A) \Impl A
\\
\False \Impl A
\\
(A \And B) \Impl A
\\
(A \And B) \Impl B
\\
A \Impl (B \Impl (A \And B))
\\
A \Impl (A \Or B)
\\
B \Impl (A \Or B)
\\
(A \Impl C) \Impl ((B \Impl C) \Impl ((A \Or B) \Impl C))
\end{eqncol}
\qquad \;\;
\begin{eqncol}
{}
\\
\Not (A \Impl B) \BImpl A \And \Not B
\\
\Not (A \And B) \BImpl \Not A \Or \Not B
\\
\Not (A \Or B) \BImpl \Not A \And \Not B
\\
\Not \Not A \BImpl A
\\
{}
\\
A \Or \Not A
\\
{}
\\
\mathbf{Rule\; of\; Inference:}
\\
\Infrule{A \quad A \Impl B}{B}
\end{eqncol}
\end{eqntbl}
\end{table}
\sloppy
In this table, $A$, $B$, and $C$ are used as meta-variables 
ranging over the set of all formulas of \LPif.
This deductive system is obtained by adding the axiom schema 
$\False \Impl A$ to the Hilbert-style deductive system of Pac given 
in~\cite{Avr91a} on page~288.
The axiom schemas on the left-hand side of the table, except for
$\False \Impl A$, together with the single inference rule (modus 
ponens) constitute a Hilbert-style deductive system for the positive 
fragment of classical propositional logic.
On the right-hand side of the table, the first three axiom schemas allow 
for the negation connective to be moved inwards, the fourth axiom schema
is the double negation axiom schema, and the fifth axiom schema is the 
law of the excluded middle.
The last-mentioned axiom can be thought of as saying that, for every 
proposition, the proposition or its negation is true, while leaving open 
the possibility that both are true.
Replacement of this axiom schema by $(A \Impl \Not A) \Impl \Not A$, 
as in CLuNs~\cite{BC04a}, yields an equivalent deductive system for 
\LPif.
Addition of the axiom schema $\Not A \Impl (A \Impl B)$, which says 
that any proposition follows from a contradiction, yields a 
Hilbert-style deductive system for classical propositional logic 
(see e.g.~\cite{Avr91a}).
The symbol $\Der$ without decoration is used to denote the derivability 
relation induced by the axiom schemas and inference rule of the given
deductive system for \LPif. 

The following outline of the semantics of \LPif\ is based 
on~\cite{Avr91a}.
Like in the case of classical propositional logic, meanings are assigned
to the formulas of \LPif\ by means of valuations.
However, in addition to the two classical truth values $\true$ (true)
and $\false$ (false), a third meaning $\both$ (both true and false) may
be assigned.
A \emph{valuation} for \LPif\ is a function $\nu$ from the set of all 
formulas of \LPif\ to the set $\{\true,\false,\both\}$ such that for all 
formulas $A$ and $B$ of \LPif:
\begin{eqnarray*}
\val{\False}{\nu} & = & \false,
\\
\val{\Not A}{\nu} & = &
 \left \{
 \begin{array}{l@{\;\;}l}
 \true  & \mathrm{if}\; \val{A}{\nu} = \false \\
 \false & \mathrm{if}\; \val{A}{\nu} = \true \\
 \both  & \mathrm{otherwise},
 \end{array}
 \right.
\\
\val{A \And B}{\nu} & = &
 \left \{
 \begin{array}{l@{\;\;}l}
 \true  & \mathrm{if}\; \val{A}{\nu} = \true  \;\mathrm{and}\;
                        \val{B}{\nu} = \true  \\
 \false & \mathrm{if}\; \val{A}{\nu} = \false \;\mathrm{or}\;
                        \val{B}{\nu} = \false \\
 \both  & \mathrm{otherwise},
 \end{array}
 \right.
\\
\val{A \Or B}{\nu} & = &
 \left \{
 \begin{array}{l@{\;\;}l}
 \true  & \mathrm{if}\; \val{A}{\nu} = \true  \;\mathrm{or}\;
                        \val{B}{\nu} = \true  \\
 \false & \mathrm{if}\; \val{A}{\nu} = \false \;\mathrm{and}\;
                        \val{B}{\nu} = \false \\
 \both  & \mathrm{otherwise},
 \end{array}
 \right.
\\
\val{A \Impl B}{\nu} & = &
 \left \{
 \begin{array}{l@{\;\;}l}
 \val{B}{\nu} & \mathrm{if}\; \val{A}{\nu} \in \set{\true,\both} \\
 \true        & \mathrm{otherwise}.
 \end{array}
 \right.
\end{eqnarray*}
The classical truth-conditions and falsehood-conditions for the logical
connectives are retained.
Except for implications, a formula is classified as both-true-and-false
exactly when it cannot be classified as true or false by the classical 
truth-conditions and falsehood-conditions.
Implications deviate in order to satisfy the standard deduction theorem.
The definition of a valuation given above shows that the logical 
connectives of \LPif\ are (three-valued) truth-functional, which means 
that each $n$-ary connective represents a function from 
$\{\true,\false,\both\}^n$ to $\{\true,\false,\both\}$.

For \LPif, the logical consequence relation, denoted by $\LCon$, is 
based on the idea that a valuation $\nu$ satisfies a formula $A$ if 
$\val{A}{\nu} \in \{\true,\both\}$.
It is defined as follows: $\Gamma \LCon A$ iff for every 
valuation $\nu$, either $\val{A'}{\nu} = \false$ for some 
$A' \in \Gamma$ or  $\val{A}{\nu} \in \{\true,\both\}$.
The given Hilbert-style deductive system for \LPif\ is sound and 
strongly complete with respect to the semantics of \LPif, i.e.\ 
$\Gamma \Der A$ iff $\Gamma \LCon A$.
This follows immediately from Theorems~1, 2, and~3 in~\cite{BC04a}.

For all formulas $A$ of \LPif\ in which $\False$ does not occur, for all 
formulas $B$ of \LPif\ in which no propositional variable occurs that 
occurs in $A$,\, $A, \Not A \nLCon B$ if $\nLCon B$ 
(cf. Proposition~4.37 in~\cite{AA15a}).%
\footnote
{On the left-hand side of $\LCon$, we write $A$ for $\set{A}$ and 
 $\Gamma,\Delta$ for $\Gamma \union \Delta$.
 Moreover, we leave out the left-hand side if it is $\emptyset$.
 We also write $\Gamma \nLCon A$ for not $\Gamma \LCon A$.}
Moreover, the connective $\Not$ satisfies the condition that, for each
propositional variable $p$, both $p \nLCon \Not p$ and 
$\Not p \nLCon p$.
Hence, \LPif\ is a paraconsistent logic.

The logical equivalence relation $\LEqv$ of \LPif\ is defined as it is 
defined for classical propositional logic: 
$A \LEqv B$ iff for every valuation $\nu$, 
$\val{A}{\nu} = \val{B}{\nu}$.
Consistency of a formula of \LPif\ is defined as follows: 
$A$ is \emph{consistent} iff for every valuation $\nu$, 
$\val{A}{\nu} \neq \both$.

Unlike in classical propositional logic, it is not the case that 
$A \LEqv B$ iff $\LCon A \BImpl B$.
Take, for example, $p \Or \Not p$ for $A$ and $q \Or \Not q$ for $B$, 
where $p$ and $q$ are different propositional variables. 
Clearly, ${} \LCon p \Or \Not p \BImpl q \Or \Not q$.
Now, let $\nu$ be a valuation such that $\val{p}{\nu} = \true$ and
$\val{q}{\nu} = \both$.
Then $\val{p \Or \Not p}{\nu} = \true$ and 
$\val{q \Or \Not q}{\nu} = \both$ and consequently it is not the case 
that $p \Or \Not p \LEqv q \Or \Not q$.
However, it is easy to check that $A \LEqv B$ only if 
${} \LCon A \BImpl B$.

\section{Known Properties of \LPif}
\label{sect-properties}

In this section, the known properties of \LPif\ that have been proposed 
as desirable properties of a reasonable paraconsistent propositional 
logic are presented.
Each of the properties in question has to do with logical consequence 
relations.
The name CL is used to denote a version of classical propositional logic 
that has the same logical constants and connectives as \LPif.
The symbol $\clLCon$ is used to denote the logical consequence relation 
of CL.

The known properties of \LPif\ that have been proposed as desirable 
properties of a reasonable paraconsistent propositional logic are:
\begin{list}{}
 {\setlength{\leftmargin}{2.4em} \settowidth{\labelwidth}{(b)}}
\item[(a)]
\emph{containment in classical logic}:
${\LCon} \subseteq {\clLCon}$;
\item[(b)]
\emph{proper basic connectives}:
for all sets $\Gamma$ of formulas of \LPif\ and all formulas $A$, $B$, 
and $C$ of \LPif:
\begin{list}{}
 {\setlength{\leftmargin}{2.7em} \settowidth{\labelwidth}{(b$_3$)}}
\item[(b$_1$)]
$\Gamma, A \LCon B$\phantom{${} \And {}$}\phantom{$C$} iff 
$\Gamma \LCon A \Impl B$,
\item[(b$_2$)]
$\Gamma \LCon A \And B$\phantom{,}\phantom{$C$} iff 
$\Gamma \LCon A$ and $\Gamma \LCon B$,
\item[(b$_3$)]
$\Gamma,  A \Or B \LCon C$ iff 
$\Gamma, A \LCon C$ and $\Gamma, B \LCon C$;
\end{list}
\item[(c)]
\emph{weak maximal paraconsistency relative to classical logic}:
for all formulas $A$ of \LPif\ with $\nLCon A$ and $\clLCon A$, for the 
minimal consequence relation $\extLCon$ such that
${\LCon} \subseteq {\extLCon}$ and $\extLCon A$, for all formulas $B$ 
of \LPif, $\extLCon B$ iff $\clLCon B$;
\item[(d)]
\emph{strongly maximal absolute paraconsistency}:
for all propositional logics $\mathcal{L}$ with the same logical 
constants and connectives as \LPif\ and a consequence relation 
$\extLCon$ such that ${\LCon} \subset {\extLCon}$, $\mathcal{L}$ is not 
paraconsistent;
\item[(e)]
\emph{internalized notion of consistency}: $A$ is consistent iff
${} \LCon (A \Impl \False) \Or (\Not A \Impl \False)$;
\item[(f)]
\emph{internalized notion of logical equivalence}: $A \LEqv B$ iff
${} \LCon (A \BImpl B) \And (\Not A \BImpl \Not B)$.
\end{list}

Properties~(a)--(f) have been mentioned relatively often in the 
literature (see e.g.~\cite{AA15a,AAZ11b,AAZ11a,Avr99a,BC04a,CCM07a}).
Properties~(a), (b$_1$), (c), and~(d) make \LPif\ an ideal 
paraconsistent logic in the sense made precise in~\cite{AAZ11b}.
By property~(e), \LPif\ is also a logic of formal inconsistency 
according to Definition~23 in~\cite{CCM07a}.%
\footnote
{The set of formulas $\bigcirc(p)$ witnessing this is 
 $\set{(p \Impl \False) \Or (\Not p \Impl \False)}$.}

Properties~(a)--(c) indicate that \LPif\ retains much of classical 
propositional logic.
In~\cite{BM15b}, properties~(e) and~(f) are considered desirable and 
essential, respectively, for a paraconsistent propositional logic on 
which a process algebra that allows for dealing with contradictory 
states is built.

From Theorem~4.42 in~\cite{AA15a}, it is known that there are exactly 
8192 different three-valued paraconsistent propositional logics with 
properties~(a) and~(b).
From Theorem~2 in~\cite{AAZ11b}, it is known that properties~(c) and~(d) 
are common properties of all three-valued paraconsistent propositional 
logics with properties~(a) and~(b$_1$).
From Fact~103 in~\cite{CCM07a}, it is known that property~(f) is a 
common property of all three-valued paraconsistent propositional logics 
with properties~(a), (b) and~(e).
Moreover, it is easy to see that property~(e) is a common property 
of all three-valued paraconsistent propositional logics with 
properties~(a) and~(b).
Hence, each three-valued paraconsistent propositional logic 
with properties~(a) and~(b) has properties~(c)--(f) as well.

From Corollary~4.74 in~\cite{AA15a}, it is known that \LPif\ is the 
strongest three-valued paraconsistent propositional logic with 
property~(a) in the sense that each three-valued paraconsistent 
propositional logic with property~(a) can be embedded in \LPif.

\section{Characterizing \LPif\ by Laws of Logical Equivalence}
\label{sect-characterization}

There are exactly 8192 different three-valued paraconsistent 
propositional logics with properties~(a) and~(b).
This means that these properties, which concern the logical consequence 
relation of a logic, have little discriminating power.
Properties~(c)--(f), which also concern the logical consequence relation 
of a logic, do not offer additional discrimination because each of the 
8192 three-valued paraconsistent propositional logics with 
properties~(a) and~(b) has these properties as well.

In this section, properties concerning the logical equivalence relation 
of a logic are used for additional discrimination.
It turns out that 11 classical laws of logical equivalence, of which at 
least 9 are considered to belong to the most basic ones, are sufficient 
to distinguish \LPif\ completely from the other 8191 three-valued 
paraconsistent propositional logics with properties~(a) and~(b).

The logical equivalence relation of \LPif\ satisfies all laws given in 
Table~\ref{laws-lequiv}.
\begin{table}[!tb]
\caption{Distinguishing laws of logical equivalence for \LPif}
\label{laws-lequiv}
\begin{eqntbl}
\begin{neqncol}
(1)  & A \And \False \LEqv \False \\
(3)  & A \And \True \LEqv A \\
(5)  & A \And A \LEqv A \\
(7)  & A \And B \LEqv B \And A \\
(9)  & \Not \Not A \LEqv A 
\end{neqncol}
\qquad\;
\begin{neqncol}
(2)  & A \Or \True \LEqv \True \\
(4)  & A \Or \False \LEqv A \\
(6)  & A \Or A \LEqv A \\
(8)  & A \Or B \LEqv B \Or A \\
(10) & \False \Impl A \LEqv \True \\
(11) & (A \Or \Not A) \Impl B \LEqv B 
\end{neqncol}
\end{eqntbl}
\end{table}
\begin{theorem}
\label{theorem-soundness}
The logical equivalence relation of \LPif\ satisfies laws (1)--(11) from 
Table~\ref{laws-lequiv}.
\end{theorem}
\begin{proof}
The proof is easy by constructing, for each of the laws concerned, truth 
tables for both sides.
\qed
\end{proof}
Moreover, among the 8192 three-valued paraconsistent propositional 
logics with properties~(a) and~(b), \LPif\ is the only one whose logical 
equivalence relation satisfies all laws given in 
Table~\ref{laws-lequiv}.
\begin{theorem}
\label{theorem-uniqueness}
There is exactly one three-valued paraconsistent propositional logic 
with properties~(a) and~(b) of which the logical equivalence relation 
satisfies laws (1)--(11) from Table~\ref{laws-lequiv}.
\end{theorem}
\begin{proof}
From Theorem~4.42 in~\cite{AA15a}, it is known that the 8192 
three-valued paraconsistent propositional logics with properties~(a) 
and~(b) are induced by a matrix of which the set of truth values is 
$\set{\true,\false,\both}$, the set of designated values is 
$\set{\true,\both}$, and the functions $\mAnd$, $\mOr$, $\mImpl$, 
and $\mNot$ on the set of truth values that correspond to the 
connectives $\And$, $\Or$, $\Impl$, and $\Not$, respectively, are 
such that, for each $b \in \set{\true,\false,\both}$:
\begin{cdispl}
\begin{eqncol}
\mAnd(\true,\true) = \true\;, \\
\mAnd(\false,b) = \false\;, \\
\mAnd(b,\false) = \false\;, \\
\mAnd(\both,\true) \in \set{\true,\both}\;, \\
\mAnd(\true,\both) \in \set{\true,\both}\;, \\
\mAnd(\both,\both) \in \set{\true,\both}\;, 
\end{eqncol}
\qquad
\begin{eqncol}
\mOr(\true,\true) = \true\;, \\
\mOr(\false,\true) = \true\;, \\
\mOr(\true,\false) = \true\;, \\
\mOr(\false,\false) = \false\;, \\
\mOr(\both,b) \in \set{\true,\both}\;, \\
\mOr(b,\both) \in \set{\true,\both}\;, 
\end{eqncol}
\qquad
\begin{eqncol}
\mImpl(\true,\true) = \true\;, \\
\mImpl(\false,\true) = \true\;, \\
\mImpl(\true,\false) = \false\;, \\
\mImpl(\false,\false) = \true\;, \\
\mImpl(\both,\false) = \false\;, \\
\mImpl(\both,\true) \in \set{\true,\both}\;, \\
\mImpl(b,\both) \in \set{\true,\both}\;, 
\end{eqncol}
\qquad
\begin{eqncol}
\mNot(\true) = \false\;, \\
\mNot(\false) = \true\;, \\
\mNot(\both) \in \set{\true,\both}\;.
\end{eqncol}
\end{cdispl}%
So, there are 8 alternatives for $\mAnd$, 32 alternatives for $\mOr$,
16 alternatives for $\mImpl$, and 2 alternatives for $\mNot$.
Below, it will be shown that, for each of these functions, laws from 
Table~\ref{laws-lequiv} exclude all but one alternative.

Law~(3) excludes $\mAnd(\both,\true) = \true$,
law~(5) excludes $\mAnd(\both,\both) = \true$, and
law~(7) excludes $\mAnd(\true,\both) = \true$.
Hence, there is only one of the 8 alternatives for $\mAnd$ left.
Law~(2) excludes $\mOr(\both,\true) = \both$,
law~(4) excludes $\mOr(\both,\false) = \true$,
law~(6) excludes $\mOr(\both,\both) = \true$, and
law~(8) excludes $\mOr(\true,\both) = \both$ and 
$\mOr(\false,\both) = \true$.
Hence, there is only one of the 32 alternatives for $\mOr$ left.
Law~(9) excludes $\mNot(\both) = \true$.
Hence, there is only one of the 2 alternatives for $\mNot$ left.
Law~(10) excludes $\mImpl(\false,\both) = \both$ and 
law~(11) excludes $\mImpl(\both,\true) = \both$,
$\mImpl(\true,\both) = \true$, and $\mImpl(\both,\both) = \true$
(in the case of the alternatives left for $\mOr$ and $\mNot$).
Hence, there is only one of the 16 alternatives for $\mImpl$ left.
\qed
\end{proof}

The following is a clarifying reformulation of the conditions on the 
functions $\mAnd$, $\mOr$, $\mImpl$, and $\mNot$ of the matrices from
the proof of Theorem~\ref{theorem-uniqueness}:
\begin{itemize}
\item
for each $\tilde{\diamond} \in \set{\mAnd,\mOr,\mImpl,\mNot}$,
the restriction of $\tilde{\diamond}$ to $\set{\true,\false}$ is
$\tilde{\diamond}$ from the matrix for classical propositional logic;
\item
for each $\tilde{\diamond} \in \set{\mAnd,\mOr,\mImpl}$,
for each $b \in \set{\true,\false}$:
\begin{cdispl}
\begin{tabular}{l}
$\tilde{\diamond}(b,\both) \neq \false$ iff
$\tilde{\diamond}(b,\true) = \true$,
\\
$\tilde{\diamond}(\both,b) \neq \false$ iff
$\tilde{\diamond}(\true,b) = \true$,
\\
$\tilde{\diamond}(\both,\both) \neq \false$ iff
$\tilde{\diamond}(\true,\true) = \true$;
\end{tabular}
\end{cdispl}
\item
$\mNot(\both) \neq \false$. 
\end{itemize}
This reformulation shows clearly that $\both$ is just an alternative for 
$\true$ in the cases of $\mAnd$, $\mOr$, and $\mImpl$, but not in the 
case of $\mNot$.

It follows immediately from the proof of 
Theorem~\ref{theorem-uniqueness} that all proper subsets of laws 
(2)--(11) from Table~\ref{laws-lequiv} are insufficient to distinguish 
\LPif\ completely from the other three-valued paraconsistent 
propositional logics with properties~(a) and~(b).
Notice that the logical equivalence relation of every three-valued 
paraconsistent propositional logics with properties~(a) and~(b) 
satisfies law~(1) from Table~\ref{laws-lequiv}.
The next two corollaries also follow immediately from the proof of 
Theorem~\ref{theorem-uniqueness}.
\begin{corollary}
\label{corollary-exactly-one}
There are exactly 16 three-valued paraconsistent propositional logics 
with properties~(a) and~(b) of which the logical equivalence relation 
satisfies laws (1)--(9) from Table~\ref{laws-lequiv}.
\end{corollary}
\begin{corollary}
\label{corollary-exactly-32}
There are exactly 32 three-valued paraconsistent propositional logics 
with properties~(a) and~(b) of which the logical equivalence relation 
satisfies laws (1)--(8) from Table~\ref{laws-lequiv}.
\end{corollary}

From a paraconsistent propositional logic with properties~(a) and~(b), 
it is only to be expected, because of paraconsistency and 
property~(b$_1$), that its negation connective and its implication 
connective deviate clearly from their counterpart in classical 
propositional logic.
Corollary~\ref{corollary-exactly-32} shows that, among the 8192 
three-valued paraconsistent propositional logics with properties~(a) 
and~(b), there are 8160 logics whose logical equivalence relation does 
not satisfy the identity, annihilation, idempotent, and commutative 
laws for conjunction and disjunction (laws (1)--(8) from 
Table~\ref{laws-lequiv}).

\section{More on the Logical Equivalence Relation of \LPif}
\label{sect-logical-equiv}

It turns out that the logical equivalence relation of \LPif\ does not 
only satisfy the identity, annihilation, idempotent, and commutative 
laws for conjunction and disjunction but also other basic classical laws 
for conjunction and disjunction, including the absorption, associative, 
distributive, and de~Morgan's laws.
Actually, the logical equivalence relation of \LPif\ also satisfies all 
laws given in Table~\ref{laws-lequiv-more}.
\begin{table}[!tb]
\caption{Additional laws of logical equivalence for \LPif}
\label{laws-lequiv-more}
\begin{eqntbl}
\begin{neqncol}
(12) & A \And (A \Or B) \LEqv A \\
(14) & (A \And B) \And C \LEqv A \And (B \And C) \\
(16) & A \And (B \Or C) \LEqv (A \And B) \Or (A \And C) \\
(18) & \Not (A \And B) \LEqv \Not A \Or \Not B \\
(20) & (A \And \Not A) \And (B \Or \Not B) \LEqv (A \And \Not A) \\
(22) & (A \Impl B) \And (A \Impl C) \LEqv A \Impl (B \And C) \\
(24) & A \Impl (B \Impl C) \LEqv (A \And B) \Impl C
\end{neqncol}
\qquad\;
\begin{neqncol}
(13) & A \Or (A \And B) \LEqv A \\
(15) & (A \Or B) \Or C \LEqv A \Or (B \Or C) \\
(17) & A \Or (B \And C) \LEqv (A \Or B) \And (A \Or C) \\
(19) & \Not (A \Or B) \LEqv \Not A \And \Not B \\
(21) & (A \And \Not A) \Or (B \Or \Not B) \LEqv (B \Or \Not B) \\
(23) & (A \Impl C) \And (B \Impl C) \LEqv (A \Or B) \Impl C 
\end{neqncol}
\end{eqntbl}
\end{table}
\begin{theorem}
\label{theorem-soundness-more}
The logical equivalence relation of \LPif\ satisfies laws (12)--(24) 
from Table~\ref{laws-lequiv-more}.
\end{theorem}
\begin{proof}
The proof is straightforward by constructing, for each of the laws 
concerned, truth tables for both sides.
\qed
\end{proof}

Laws (1)--(9) and (12)--(21) axiomatize bounded normal 
i-lattices~\cite{Kal58a}.%
\footnote
{Bounded normal i-lattices are also (confusingly) called Kleene algebras
 (see e.g.~\cite{BM67a}).}
Laws~(10)--(11) and (22)--(24) are laws concerning the implication 
connective. 
Laws (10)--(11) and (22)--(24), like laws (1)--(9) and (12)--(21), are 
also satisfied by the logical equivalence relation of classical 
propositional logic. 
For all formulas $A'$ and $A''$ of classical propositional logic, the 
logical equivalence relation of classical propositional logic satisfies 
$A' \LEqv A''$ iff $A' \LEqv A''$ follows from laws (1)--(9) and 
(12)--(21) and the laws
\begin{ldispl}
(25) \;\; A \And \Not A \LEqv \False \qquad\;  
(26) \;\; A \Or  \Not A \LEqv \True  \qquad\;
(27) \;\; A \Impl B \LEqv \Not A \Or B \;.\footnotemark   
\end{ldispl}%
\footnotetext
{This fact is easy to see because, without law (27), these laws
 axiomatize Boolean algebras and, in classical propositional logic, law
 (27) defines $\Impl$ in terms of $\Or$ and $\Not$.}%
However, laws (25)--(27) are not satisfied by the logical equivalence 
relation of \LPif. 

The fact that $A \LEqv B$ iff
${} \LCon (A \BImpl B) \And (\Not A \BImpl \Not B)$
suggests that \LPif\ is Blok-Pigozzi algebraizable~\cite{BP89a}.
\begin{theorem}
\label{theorem-algebraizable}
\LPif\ is finitely Blok-Pigozzi algebraizable with the equivalence 
formulas
$\set
  {p \Impl q, q \Impl p, \Not p \Impl \Not q, \Not q \Impl \Not p}$
and the single defining equation $p = p \Or \Not p$ ($p$ and $q$ are 
propositional variables).
\end{theorem}
\begin{proof}
Because 
$A \LEqv B$ iff $\,\LCon (A \BImpl B) \And (\Not A \BImpl \Not B)$,
$A \BImpl B$ stands for $(A \Impl B) \And\linebreak[2] (B \Impl A)$, 
and $\LCon A \And B$ iff $\LCon A$ and $\LCon B$, it is the case that
$A \LEqv B$ iff $\LCon A \Impl B$ and $\LCon B \Impl A$ and
$\LCon \Not A \Impl \Not B$ and $\LCon \Not B \Impl \Not A$. 
Moreover, writing 
$A \ILEqv B$ for $(A \BImpl B) \And (\Not A \BImpl \Not B)$, 
it is easily found that
(i)~$\LCon A \ILEqv A$, 
(ii)~$A, A \ILEqv B \LCon B$,
(iii)~$A \ILEqv B \LCon \Not A \ILEqv \Not B$,  
(iv)~$A \ILEqv B, A' \ILEqv B' \LCon A \diamond A' \ILEqv B \diamond B'$
     for $\diamond \in \set{\And, \Or, \Impl}$, and
(v)~$A \LCon A \ILEqv A \Or \Not A$ and 
    $A \ILEqv A \Or \Not A \LCon A$.
Therefore, by Corollary~3.6 from~\cite{Her96a} and the fact that
\LPif\ is a finitary logic, \LPif\ is Blok-Pigozzi algebraizable%
\footnote
{In~\cite{Her96a}, Blok-Pigozzi algebraizable is called finitely
 algebraizable.}
with the equivalence formulas
$\set
  {p \Impl q, q \Impl p, \Not p \Impl \Not q, \Not q \Impl \Not p}$
and the single defining equation $p = p \Or \Not p$.
\qed
\end{proof}
The algebraization concerned is the quasi-variety generated by the 
expansion of the 3-element bounded normal i-lattice obtained by adding 
the unique binary operation $\Impl$ that satisfies 
$\False \Impl p = \True$ and $(p \Or \Not p) \Impl q = q$. 

\section{The Paracomplete Analogue of \LPif}
\label{sect-dual}

In this section, the paracomplete analogue of \LPif\ is introduced.
In Section~\ref{sect-dual-properties}, properties of this logic are 
presented that are comparable to properties of \LPif\ that have been 
presented in the preceding sections.

Replacing the axiom schema $A \Or \Not A$ by the axiom schema
$\Not A \Impl (A \Impl B)$  in the given Hilbert-style deductive 
system for \LPif\ yields a Hilbert-style deductive system for Kleene's 
strong three-valued logic introduced in Section~64 of~\cite{Kle52a} with 
its implication connective replaced by an implication connective for 
which the standard deduction theorem holds and enriched with a falsity 
constant.
The name \KLif\ is used to denote this logic.
It is perhaps clarifying that the axiom schemas involved in the 
above-mentioned replacement can be paraphrased as ``$A$ or $\Not A$ 
follows from anything'' and ``anything follows from $A$ and $\Not A$'', 
respectively.
Virtually all differences between \LPif\ and \KLif\ can be traced to the 
fact that the third truth value $\both$ is interpreted as both true and 
false in the former logic and as neither true nor false in the latter 
logic.

Like in the case of \LPif, meanings are assigned to the formulas of 
\KLif\ by means of valuations that are functions from the set of all 
formulas of \KLif\ to the set $\{\true,\false,\both\}$.
The conditions that a valuation for \KLif\ must satisfy differ from the 
conditions that a valuation for \LPif\ must satisfy only with respect to 
implication:
\begin{eqnarray*}
\val{A \Impl B}{\nu} & = &
 \left \{
 \begin{array}{l@{\;\;}l}
 \val{B}{\nu} & \mathrm{if}\; \val{A}{\nu} = \true \\
 \true        & \mathrm{otherwise}.
 \end{array}
 \right.
\end{eqnarray*}

The conditions that a valuation for \KLif\ must satisfy uniquely 
characterize the set of valuations induced by the truth value that 
naturally corresponds to the falsity constant and the functions on 
the set of truth values that correspond to the different connectives 
according to Section~3.1 of~\cite{Avr91a}.
The functions concerned, except the function that corresponds to the 
implication connective, are the same as the ones that are represented by 
the truth tables presented in Section~64 of~\cite{Kle52a}.

The symbol $\Der^\star$ is used to denote the derivability relation 
induced by the axiom schemas and inference rule of the deductive system
for \KLif\ referred to above.
The logical consequence relation of \KLif, denoted by $\LCon^\star$, is 
based on the idea that a valuation $\nu$ satisfies a formula $A$ if 
$\val{A}{\nu} = \true$.
It is defined as follows: $\Gamma \LCon^\star A$ iff for every 
valuation $\nu$, either $\val{A'}{\nu} \in \{\false,\both\}$ for some 
$A' \in \Gamma$ or $\val{A}{\nu} = \true$.
The Hilbert-style deductive system for \KLif\ referred to above is sound 
and strongly complete with respect to the semantics of \KLif, i.e.\ 
$\Gamma \Der^\star A$ iff $\Gamma \LCon^\star A$.
This is proved for the fragment of \KLif\ without the falsity constant 
in Section~4 of~\cite{Avr91a}. 
The proof concerned easily generalizes to full \KLif.

It follows immediately from the definition of $\LCon^\star$ that, for 
all formulas $A$ and $B$ of \KLif, $A, \Not A \LCon^\star B$.
Hence, \KLif\ is not a paraconsistent propositional logic.
\KLif\ is a paracomplete propositional logic instead.

A propositional logic $\mathcal{L}$ is a \emph{paracomplete} 
propositional logic if 
(a)~its logical consequence relation $\LCon_{\mathcal{L}}$ satisfies the 
condition that there exists a formula $A$ of $\mathcal{L}$ such that 
$\nLCon_{\mathcal{L}} A \Or \Not A$ and
(b)~its negation connective $\Not$ satisfies the condition that, 
for each propositional variable $p$, both 
$p \nLCon_{\mathcal{L}} \Not p$ and $\Not p \nLCon_{\mathcal{L}} p$.

The logical equivalence relation $\LEqv$ of \KLif\ is defined as it is 
defined for \LPif\ and classical propositional logic: $A \LEqv B$ iff 
for every valuation $\nu$, $\val{A}{\nu} = \val{B}{\nu}$.
Definedness of a formula of \KLif\ is defined as consistency of a 
formula of \LPif\ is defined: $A$ is \emph{defined} iff for every 
valuation $\nu$, $\val{A}{\nu} \neq \both$.

\KLif\ is essentially the same as the propositional fragment of LPF 
(Logic of Partial Functions)~\cite{BCJ84a,Che86a} without the constant 
representing $\both$ (cf.\ Section~3.1.2 of~\cite{Avr91a}).

It is a notable fact that \KLif\ and \LPif\ are dual in the way that 
classical logical consequences are reflected in these logics.
Below, this is made more precise.

We write $V(A)$, where $A$ is a formula of CL, for the set of all 
propositional variables occurring in $A$.

Let $p_1,\ldots,p_n$, where $n \geq 1$, be propositional variables, and 
let $A$ be a formula of CL such that $V(A) = \set{p_1,\ldots,p_n}$.
Then the \emph{inconsistency formula} for $A$, written $\iota(A)$, is 
defined by
$\iota(A) =
 (p_1 \And \Not p_1) \Or \ldots \Or (p_n \And \Not p_n)$ 
and the \emph{definedness formula} for $A$, written $\delta(A)$, is 
defined by
$\delta(A) =
 (p_1 \Or \Not p_1) \And \ldots \And (p_n \Or \Not p_n)$.
The convention is used that $\iota(A) = \False$ and $\delta(A) = \True$ 
if $V(A) = \emptyset$.
Clearly, $\LCon \iota(A) \BImpl \Not \delta(A)$ and 
$\LCon^\star \delta(A) \BImpl \Not \iota(A)$.

\begin{theorem}
\label{theorem-duality}
For all formulas $A$ and $B$ of CL:%
\footnote{Clearly, the formulas of CL, \LPif\ and \KLif\ are the same.} 
\begin{ldispl}
A \clLCon B \;\;\mathrm{iff}\;\; A \LCon B \Or \iota(A) \\
\phantom{A \clLCon B \;} \mathrm{and} \\ 
A \clLCon B \;\;\mathrm{iff}\;\; A \And \delta(B) \LCon^\star B\;.
\end{ldispl}%
\end{theorem}
\begin{proof}
Theorems~3.10 and~4.2 in~\cite{Bea13a} constitute a similar result for
the first-order extensions of the $\set{\Not,\And,\Or}$-fragments of
\LPif\ and \KLif\ in a setting with multiple-conclusion logical 
consequence relations.
The proofs of these theorems do not depend on the absence of $\Impl$ 
and $\False$ and can be adapted with minor changes to the propositional 
case and single-conclusion logical consequence relations.
\qed
\end{proof}
Theorem~\ref{theorem-duality} can be generalized to the case with 
multiple premises. 
However, in that case, the symmetry that shows up in the theorem gets 
lost in the single-conclusion versions of \LPif\ and \KLif\ presented in 
this paper.

\section{Properties of \KLif\ Comparable to Properties of \LPif}
\label{sect-dual-properties}

In this section, properties of \KLif\ are presented that are comparable 
to properties of \LPif\ that have been proposed as desirable properties 
of a reasonable paraconsistent propositional logic and a property of 
\KLif\ is proved that is comparable to the property of \LPif\ concerning 
the logical equivalence relation that distinguishes \LPif\ from other 
reasonable three-valued paraconsistent propositional logics.

The following properties of \KLif\ are similar to properties of \LPif\ 
mentioned in Section~\ref{sect-properties}:
\begin{list}{}
 {\setlength{\leftmargin}{2.4em} \settowidth{\labelwidth}{(b$'$)}}
\item[(a$'$)]
\emph{containment in classical logic}:
${\LCon^\star} \subseteq {\clLCon}$;
\item[(b$'$)]
\emph{proper basic connectives}:
for all sets $\Gamma$ of formulas of \KLif\ and all formulas $A$, $B$, 
and $C$ of \KLif:
\begin{list}{}
 {\setlength{\leftmargin}{2.7em} \settowidth{\labelwidth}{(b$'_3$)}}
\item[(b$'_1$)]
$\Gamma, A \LCon^\star B$\phantom{${} \And {}$}\phantom{$C$} iff 
$\Gamma \LCon^\star A \Impl B$,
\item[(b$'_2$)]
$\Gamma \LCon^\star A \And B$\phantom{,}\phantom{$C$} iff 
$\Gamma \LCon^\star A$ and $\Gamma \LCon^\star B$,
\item[(b$'_3$)]
$\Gamma,  A \Or B \LCon^\star C$ iff 
$\Gamma, A \LCon^\star C$ and $\Gamma, B \LCon^\star C$;
\end{list}
\item[(c$'$)]
\emph{weakly maximal paracompleteness relative to classical logic}:
for all formulas $A$ of \KLif\ with $\nLCon^\star A$ and $\clLCon A$, 
for the minimal consequence relation $\extLCon$ such that
${\LCon^\star} \subseteq {\extLCon}$ and $\extLCon A$, for all formulas 
$B$ of \KLif, $\extLCon B$ iff $\clLCon B$;
\item[(e$'$)]
\emph{internalized notion of definedness}: $A$ is defined iff
${}\LCon^\star \Not (A \Impl \False) \Or \Not (\Not A \Impl \False)$;
\item[(f$'$)]
\emph{internalized notion of logical equivalence}: $A \LEqv B$ iff
${} \LCon^\star (A \BImpl B) \And (\Not A \BImpl \Not B)$.
\end{list}
It is easy to see that \KLif\ has properties~(a$'$), (b$'$), (e$'$), 
and~(f$'$).
Property~(c$'$) is stated in Section~5 of~\cite{BCK99a} (where \KLif\ is 
called CLaNs).
It remains an open question whether \KLif\ has the following property:
\begin{list}{}
 {\setlength{\leftmargin}{2.4em} \settowidth{\labelwidth}{(d$'$)}}
\item[(d$'$)]
\emph{strongly maximal absolute paracompleteness}:
for all propositional logics $\mathcal{L}$ with the same logical 
constants and connectives as \KLif\ and a consequence relation 
$\extLCon$ such that ${\LCon^\star} \subset {\extLCon}$, $\mathcal{L}$ 
is not paracomplete.
\end{list}

By property~(e$'$), \KLif\ is a logic of formal undeterminedness in the 
sense made precise in~\cite{CCR19a}.

The following theorem concerns the number of three-valued paracomplete 
propositional logics with properties~(a$'$) and~(b$'$).
\begin{theorem}
\label{theorem-1024}
There are exactly 1024 three-valued paracomplete propositional logics 
with properties~(a$'$) and~(b$'$).%
\footnote
{It seems that \cite{Mar05a} refers to an unpublished document about 
 these 1024 paracomplete propositional logics.}
\end{theorem}
\begin{proof}
Consider the set of 1024 three-valued paracomplete propositional logics 
that are induced by a matrix of which the set of truth values is 
$\set{\true,\false,\both}$, the set of designated values is 
$\set{\true}$, and the functions $\mAnd$, $\mOr$, $\mImpl$, 
and $\mNot$ on the set of truth values that correspond to the 
connectives $\And$, $\Or$, $\Impl$, and $\Not$, respectively, are 
such that, for each $b \in \set{\true,\false,\both}$:
\begin{cdispl}
\begin{eqncol}
\mAnd(\true,\true) = \true\;, \\
\mAnd(\false,\true) = \false\;, \\
\mAnd(\true,\false) = \false\;, \\
\mAnd(\false,\false) = \false\;, \\
\mAnd(\both,b) \in \set{\false,\both}\;, \\
\mAnd(b,\both) \in \set{\false,\both}\;, 
\end{eqncol}
\qquad
\begin{eqncol}
\mOr(\true,b) = \true\;, \\
\mOr(b,\true) = \true\;, \\
\mOr(\false,\false) = \false\;, \\
\mOr(\both,\false) \in \set{\false,\both}\;, \\
\mOr(\false,\both) \in \set{\false,\both}\;, \\
\mOr(\both,\both) \in \set{\false,\both}\;, 
\end{eqncol}
\qquad
\begin{eqncol}
\mImpl(\true,\true) = \true\;, \\
\mImpl(\true,\false) = \false\;, \\
\mImpl(\false,b) = \true\;, \\
\mImpl(\both,b) = \true\;, \\
\mImpl(\true,\both) \in \set{\false,\both}\;,
\end{eqncol}
\qquad
\begin{eqncol}
\mNot(\true) = \false\;, \\
\mNot(\false) = \true\;, \\
\mNot(\both) \in \set{\false,\both}\;.
\end{eqncol}
\end{cdispl}%
For an induced logic,
$\nu(A \diamond B) = \tilde{\diamond}(\nu(A),\nu(B))$
for each $\diamond \in \set{\And,\Or,\Impl}$ and 
$\nu(\Not A) = \mNot(\nu(A))$.
From this it follows easily that the above conditions on the functions 
$\mAnd$, $\mOr$, $\mImpl$, and $\mNot$ are equivalent to the 
condition that, for each valuation~$\nu$:
\begin{cdispl}
\begin{tabular}{l}
$\nu(A \And B) = \true$  iff $\nu(A) = \true$    and $\nu(B) = \true$,
\\ 
$\nu(A \Or B) = \true$   iff $\nu(A) = \true$    or  $\nu(B) = \true$,
\\ 
$\nu(A \Impl B) = \true$ iff $\nu(A) \neq \true$ or  $\nu(B) = \true$,
\\ 
$\nu(\Not A) = \true$\phantom{$\;B\;$} iff 
$\nu(A) \neq \true$ and $\nu(A) \neq \both$.
\end{tabular}
\end{cdispl}%
For an induced logic, $\Gamma \LCon^\star A$ iff for every valuation 
$\nu$, either $\val{A'}{\nu} \neq \true$ for some $A' \in \Gamma$ or 
$\val{A}{\nu} = \true$.
From this it follows easily that the above condition on valuations 
implies properties~(a$'$) and~(b$'$).
Moreover, it follows immediately from the above-mentioned properties of 
an induced logic, that $\mNot(\both) \in \set{\false,\both}$ iff the 
induced logic is paracomplete.

In other words, by the above conditions on the functions $\mAnd$, 
$\mOr$, $\mImpl$, and $\mNot$, the 1024 induced logics are 
three-valued paracomplete logics that have properties~(a$'$) and~(b$'$).
It is still to be shown that there are no other three-valued 
paracomplete logics that has properties~(a$'$) and~(b$'$). 

The only weakening of the above condition on valuations that still 
guarantees property (a$'$) is the one that is obtained by replacing
\begin{cdispl}
\begin{tabular}{l}
$\nu(\Not A) = \true$ iff $\nu(A) \neq \true$ and $\nu(A) \neq \both$
$\quad$ by $\quad$
$\nu(\Not A) = \true$ iff $\nu(A) \neq \true$.
\end{tabular}
\end{cdispl}%
This weakening is equivalent to the weakening of the above conditions on 
the functions $\mAnd$, $\mOr$, $\mImpl$, and $\mNot$ that is obtained
by replacing
\begin{cdispl}
\begin{tabular}{l}
$\mNot(\both) \in \set{\false,\both}$
$\quad$ by $\quad$
$\mNot(\both) \in \set{\true,\false,\both}$.
\end{tabular}
\end{cdispl}%
However, if $\mNot(\both) = \true$, then the induced logic is not
paracomplete.
\qed
\end{proof}

The following is a clarifying reformulation of the conditions on the 
functions $\mAnd$, $\mOr$, $\mImpl$, and $\mNot$ of the matrices from
the proof of Theorem~\ref{theorem-1024}:
\begin{itemize}
\item
for each $\tilde{\diamond} \in \set{\mAnd,\mOr,\mImpl,\mNot}$,
the restriction of $\tilde{\diamond}$ to $\set{\true,\false}$ is
$\tilde{\diamond}$ from the matrix for classical propositional logic;
\item
for each $\tilde{\diamond} \in \set{\mAnd,\mOr,\mImpl}$,
for each $b \in \set{\true,\false}$:
\begin{cdispl}
\begin{tabular}{l}
$\tilde{\diamond}(b,\both) \neq \true$ iff
$\tilde{\diamond}(b,\false) = \false$,
\\
$\tilde{\diamond}(\both,b) \neq \true$ iff
$\tilde{\diamond}(\false,b) = \false$,
\\
$\tilde{\diamond}(\both,\both) \neq \true$ iff
$\tilde{\diamond}(\false,\false) = \false$;
\end{tabular}
\end{cdispl}
\item
$\mNot(\both) \neq \true$. 
\end{itemize}
This reformulation shows clearly that $\both$ is just an alternative for 
$\false$ in the cases of $\mAnd$, $\mOr$, and $\mImpl$, but not in the 
case of $\mNot$.

The next theorem concerns the logical equivalence relation of \KLif.
\begin{theorem}
\label{theorem-soundness-dual}
The logical equivalence relation of \KLif\ satisfies laws (1)--(9) from 
Table~\ref{laws-lequiv} and the laws (10\/$'$) $\True \Impl A \LEqv A$ 
and (11\/$'$) $(A \And \Not A) \Impl B \LEqv \True$.
\end{theorem}
\begin{proof}
The proof is easy by constructing, for each of the laws concerned, truth 
tables for both sides.
For laws (1)--(9), this proof coincides with the proof of 
Theorem~\ref{theorem-soundness}.
\qed
\end{proof}
The 11 laws of logical equivalence from 
Theorem~\ref{theorem-soundness-dual} are sufficient to distinguish 
\KLif\ completely from the other 1023 three-valued paracomplete 
propositional logics with properties~(a$'$) and~(b$'$).
\begin{theorem}
\label{theorem-uniqueness-dual}
There is exactly one three-valued paracomplete propositional logic 
with properties~(a$'$) and~(b$'$) of which the logical equivalence 
relation satisfies laws (1)--(9) from Table~\ref{laws-lequiv} and the 
laws (10\/$'$) $\True \Impl A \LEqv A$ and 
(11\/$'$) $(A \And \Not A) \Impl B \LEqv \True$.
\end{theorem}
\begin{proof}
The proof is similar to the proof of Theorem~\ref{theorem-uniqueness}.
The matrices involved are now the 1024 matrices from the proof of 
Theorem~\ref{theorem-1024}.
In this case, there are 32 alternatives for $\mAnd$, 8 alternatives for 
$\mOr$, 2 alternatives for $\mImpl$, and 2 alternatives for~$\mNot$.
Below, it will be shown that, for each of these functions, laws from the 
ones mentioned in the theorem exclude all but one alternative.

Law~(1) excludes $\mAnd(\both,\false) = \both$,
law~(3) excludes $\mAnd(\both,\true) = \false$,
law~(5) excludes $\mAnd(\both,\both) = \false$, and
law~(7) excludes $\mAnd(\false,\both) = \both$ and 
$\mAnd(\true,\both) = \false$.
Hence, there is only one of the 32 alternatives for $\mAnd$ left.
Law~(4) excludes $\mOr(\both,\false) = \false$,
law~(6) excludes $\mOr(\both,\both) = \false$, and
law~(8) excludes $\mOr(\false,\both) = \false$.
Hence, there is only one of the 8 alternatives for $\mOr$ left.
Law~(9) excludes $\mNot(\both) = \false$.
Hence, there is only one of the 2 alternatives for $\mNot$ left.
Law (10$'$) excludes $\mImpl(\true,\both) = \false$.
Hence, there is only one of the 2 alternatives for $\mImpl$ left.
\qed
\end{proof}

\begin{theorem}
\label{theorem-soundness-more-dual}
The logical equivalence relation of \KLif\ satisfies laws (12)--(24) 
from Table~\ref{laws-lequiv-more}.
\end{theorem}
\begin{proof}
Because \LPif\ and \KLif\ have the same truth tables for conjunction, 
disjunction and negation, the logical equivalence relation of \KLif\ 
also satisfies laws~(12)--(21) from Table~\ref{laws-lequiv-more}. 
Proving that the logical equivalence relation of \KLif\ also satisfies 
laws~(22)--(24) from Table~\ref{laws-lequiv-more} is straightforward by 
constructing, for each of the laws concerned, truth tables for both 
sides.
\qed
\end{proof}

Like \LPif, \KLif\ is Blok-Pigozzi algebraizable.
\begin{theorem}
\label{theorem-algebraizable-dual}
\KLif\ is finitely Blok-Pigozzi algebraizable with the equivalence 
formulas
$\set
  {p \Impl q, q \Impl p, \Not p \Impl \Not q, \Not q \Impl \Not p}$
and the single defining equation $p = p$ ($p$ and $q$ are propositional 
variables).
\end{theorem}
\begin{proof}
The proof follows the same line as the proof of 
Theorem~\ref{theorem-algebraizable}.
\qed
\end{proof}
The algebraization concerned is the quasi-variety generated by the 
expansion of the 3-element bounded normal i-lattice obtained by adding 
the unique binary operation $\Impl$ that satisfies 
$\True \Impl p = p$ and $(p \And \Not p) \Impl q = \True$. 

\section{Concluding Remarks}
\label{sect-concl}

In this paper, properties concerning the logical equivalence relation of 
a logic are used to distinguish the logic \LPif\ from the other logics 
that belong to the 8192 three-valued paraconsistent propositional logics 
that have properties (a)--(f) from Section~\ref{sect-properties}.
These 8192 logics are regarded as potentially interesting because  
properties (a)--(f) are generally considered to be desirable properties 
of a reasonable paraconsistent propositional logic.

Properties (a)--(f) concern the logical consequence relation of a logic.
Unlike in classical propositional logic, it is not the case that 
$A \LEqv B$ iff $A \LCon B$ and $B \LCon A$ in a three-valued 
paraconsistent propositional logic.
As a consequence, the classical laws of logical equivalence that follow 
from property~(b) in classical propositional logic, viz.\ laws (1)--(8) 
and (12)--(17) from Section~\ref{sect-characterization}, do not follow 
from property~(b) in a three-valued paraconsistent propositional logic.
Therefore, if closeness to classical propositional logic is considered 
important, it should be a desirable property of a reasonable 
paraconsistent propositional logic to have a logical equivalence 
relation that satisfies laws (1)--(8) and  (12)--(17). \linebreak[2]
This would reduce the potentially interesting three-valued 
paraconsistent propositional logics from 8192 to 32.

In~\cite{BM15b}, satisfaction of laws (1)--(8), (11), (14)--(17), and 
(22)--(24) is considered essential for a paraconsistent propositional 
logic on which a process algebra that allows for dealing with 
contradictory states is built.
It follows easily from Theorem~\ref{theorem-soundness} and the proof of 
Theorem~\ref{theorem-uniqueness} that \LPif\ is one of only four 
three-valued paraconsistent propositional logics with properties (a) and 
(b) of which the logical equivalence relation satisfies laws (1)--(8), 
(11), (14)--(17), and \linebreak[2] (22)--(24).

It is also shown in this paper that, for most presented properties of 
\LPif, its paracomplete analogue \KLif\ has a comparable property.

\subsection*{Acknowledgements}

We thank three anonymous referees for carefully reading preliminary 
versions of this paper, for pointing out several flaws in it, and for 
suggesting improvements of the presentation.

\bibliographystyle{splncs03}
\bibliography{PCL}

\end{document}